\documentclass[copyright,creativecommons,noderivs]{eptcs}
\providecommand{\event}{PDMC 2011}

\usepackage{amssymb}
\usepackage{url}
\usepackage{algorithm2e}
\usepackage{color}
\usepackage{breakurl}
\usepackage{amsthm}
\usepackage{rotating}
\usepackage{graphicx}

\graphicspath{{data/}}

\usepackage{cite}
\usepackage[caption=false,font=normalsize]{subfig}

\newtheorem{theorem}{Theorem}
\newtheorem{lemma}{Lemma}
\newtheorem{corollary}{Corollary}
\newtheorem{definition}{Definition}

\SetKwInOut{Input}{Input}
\SetKwInOut{Output}{Output}

\title{Lazy Decomposition for Distributed Decision Procedures}

\author{
Youssef Hamadi\institute{Microsoft Research\\Cambridge, UK}\email{youssefh@microsoft.com} \and
Joao Marques-Silva\institute{University College\\Dublin, IE}\email{jpms@ucd.ie} \and
Christoph M. Wintersteiger\institute{Microsoft Research\\Cambridge, UK}\email{cwinter@microsoft.com}
}

\def\titlerunning{Lazy Decomposition for Distributed Decision Procedures}
\def\authorrunning{Y. Hamadi, J. Marques-Silva, and C.M. Wintersteiger}

\begin{document}
\maketitle


\begin{abstract}
The increasing popularity of automated tools for software and hardware
verification puts ever increasing demands on the underlying decision
procedures. This paper presents a framework for distributed decision procedures
(for first-order problems) based on Craig interpolation. Formulas are
distributed in a lazy fashion, i.e., without the use of costly decomposition
algorithms. Potential models which are shown to be incorrect are reconciled
through the use of Craig interpolants. Experimental results on challenging
propositional satisfiability problems indicate that our method is able to 
outperform traditional solving techniques even without the use of additional
resources.
\end{abstract}


\section{Introduction}

Decision procedures for first-order logic problems, or fragments thereof, have
seen a tremendous increase in popularity in the recent past. This is due to the
great increase in performance of solvers for the propositional satisfiability
(SAT) problem, as well as the increasing popularity of verification tools for
both soft- and hardware which extensively use first-order decision procedures
like SAT and SMT solvers.

As the decision problems that occur in large-scale verification problems become
larger, methods for distribution and parallelization are required to overcome
memory and runtime limitations of modern computer systems. Frequently,
computing clusters and multi-core processors are employed to solve such
problems. The inherent parallelism in these systems is often used to solve
multiple problems concurrently, while distributed and parallel decision
procedures would allow for much better performance. This has led to the
development of distributed verification tools, for example through the
distribution of Bounded-Model-Checking (BMC) problems (see,
e.g.,~\cite{GanaiGYA06,CamposNZS09}).

In the following sections we present a general method for distributed decision
procedures which is applicable to decision procedures of first-order
fragments. The key component in this method is Craig's interpolation
theorem~\cite{Craig57}. This theorem enables us to arbitrarily split formulas
into multiple parts without any restrictions on the nature or size of the
cut. We propose to use this lazy formula decomposition because it does not
require analysis of the semantics of a formula prior to the distribution of the
problem. In many other distributed algorithms, such a lazy decomposition
clearly has a negative impact on the overall runtime of the decision
procedure. However, when using an interpolation scheme, the abstraction
provided by the interpolation algorithm is often strong enough to
counterbalance this.

Through an experimental evaluation of our algorithm on propositional formulas,
we are able to show that, at least in the propositional case, large speed-ups
result from using a lazy decomposition when using a suitable interpolation
algorithm. 


\section{Background}

We are interested in satisfiability of first-order formulas. Formulas are
assumed to be in conjunctive normal form, i.e., of the form 
$\phi = \psi_1(x_1,\dots,x_n)\wedge\dots\wedge\psi_m(x_1,\dots x_n),$ where 
$V_\phi = \{x_1,\dots,x_n\}$ are the free variables of $\phi$ and the $\psi_i$
are clauses (disjunctions of atoms).

Craig's interpolation theorem provides a way to characterize the relationship
between two formulas when one implies the other:

\begin{theorem}[\label{thm:craig}Craig Interpolation~\cite{Craig57}]
Let $\phi$ and $\psi$ be first-order formulas. If $\phi\Rightarrow\psi$
then there exists an \emph{Interpolant} $I$ such that $\phi\Rightarrow I 
\;\wedge\; I\Rightarrow\psi$ and $V_I\subseteq V_\phi\cap V_\psi$.
\end{theorem}

Equivalently, there is an interpolant $I$ such that 
$\phi\Rightarrow I \;\wedge\; I\Rightarrow\neg\psi$
whenever $\phi\wedge\psi$ is unsatisfiable, because
$\phi\Rightarrow\neg\psi\equiv \neg(\phi\wedge\psi)$.
Craig's theorem guarantees the existence of an interpolant, but does not
provide an algorithm for obtaining it. However, such algorithms are known for
many logics. We refer to two interpolation algorithms for propositional logic
and to describe them, we require some definitions: 

A literal $l$ is either a propositional variable $x$ or its negation $\neg
x$. A clause is a disjunction of literals, denoted $\{ l_1, \dots, l_n \}$. A
formula $\phi$ is assumed to be in conjunctive normal form (CNF), i.e., it is a
conjunction of clauses, denoted $\phi = \{ c_1, \dots, c_n \}$. 

A (partial) assignment $\alpha$ is a consistent set of clauses of size 1. The
(propositional) SAT problem is to determine whether for a given formula $\phi$
there exists a total assignment such that $\phi \wedge \alpha\equiv\top$. 
Given two clauses of the form $C_1 \cup \{x\}$ and $C_2 \cup \{\neg x\}$, 
their resolution is defined as $C_1\cup C_2$ and if it is not tautological the
result is called a \emph{resolvent} and the variable $x$ is called the pivot
variable. A resolution refutation is a sequence of resolution operations such
that the final resolvent is empty, proving that the formula is unsatisfiable.

There are multiple techniques for propositional interpolation. Here,
we refer to two popular systems, one by McMillan~\cite{McMillan03} and the
other by Huang, Kraj\'icek and Pudl\'ak~\cite{Huang95,Krajicek97,Pudlak97}
(HKP). Both are methods that require time linear in the size of the resolution
refutation of $\neg(\phi\wedge\psi)$. Both interpolation methods construct
interpolants by associating an intermediate interpolant
with every resolvent in the resolution refutation of $\phi\wedge\psi$. The
interpolant associated with the final resolvent (the empty clause) constitutes
an interpolant $I$ for which Theorem~\ref{thm:craig} holds. For a
characterization of these and other interpolation systems see,
e.g.,~\cite{DSilva10}.

McMillan's interpolation system associates with every clause
$C$ in $\phi$ the intermediate interpolant $C \setminus \{ v, \neg v | v \in
V_\psi\}$, i.e., the restriction of $C$ to the variables in $\psi$. Every
clause in $\psi$ is associated with the intermediate interpolant $\top$. Every other
interpolant is calculated depending on the corresponding resolution
step. Consider the derivation of resolvent $R$ from clauses $C_1$ and $C_2$
with pivot variable $x$, where $C_1$ and $C_2$ have previously been associated
with intermediate interpolants $I_{C_1}$ and $I_{C_2}$. The resulting clause $R$
is then associated with the intermediate interpolant
\begin{displaymath}
I_R := \left\{
\begin{array}{cl}
  I_{C_1} \vee I_{C_2} & \mbox{if } x\in V_\phi \wedge x\not\in V_\psi\\
  I_{C_1} \wedge I_{C_2} & \mbox{if } x\in V_\phi \wedge x\in V_\psi\\
  I_{C_1} \wedge I_{C_2} & \mbox{if } x\not\in V_\phi \wedge x\in V_\psi
\end{array}
\right.
\;.
\end{displaymath}

In the HKP system, every clause in $\phi$ is associated the intermediate
interpolant $\bot$, while the clauses in $\psi$ are associated $\top$. Every
resolvent $R$ obtained from clauses $C_1$ and $C_2$ with pivot variable $x$ is
associated with the intermediate interpolant 
\begin{displaymath}
I_R := \left\{
\begin{array}{cl}
  I_{C_1} \vee I_{C_2} & \mbox{if } x\in V_\phi \wedge x\not\in V_\psi\\
  (x \vee I_{C_1}) \wedge (\neg x \vee I_{C_2}) & \mbox{if } x\in V_\phi \wedge x\in V_\psi\\
  I_{C_1} \wedge I_{C_2} & \mbox{if } x\not\in V_\phi \wedge x\in V_\psi
\end{array}
\right.
\;.
\end{displaymath}

For every propositional interpolation system that computes an interpolant
for $\phi\Rightarrow\psi$, the \emph{dual} system is defined by the computation
of an interpolant for $I$ for $\psi\Rightarrow\phi$, with the effect that 
$\neg I$ is an interpolant for $\phi\Rightarrow\psi$. It is known that the HKP
system is self-dual and that McMillan's system is not~\cite{DSilva10}.


\section{Related Work}

Our work is most closely related to parallel and distributed decision
procedures. Many decision procedures exists for the propositional
satisfiability and some of them exploit parallelism. 

\subsection{Parallel SAT Solving}

In parallel SAT, the objective is to simultaneously explore different parts of
the search space in order to quickly solve a problem. There are two main
approaches to parallel SAT solving. First, the classical concept of
divide-and-conquer, which divides the search space into subspaces and allocate
each of them to sequential SAT solvers. The search space is divided thanks to 
guiding-path constraints (typically unit clauses). A formula is found
satisfiable if one worker is able to find a solution for its subspace, and
unsatisfiable if all the subspaces have been proved unsatisfiable. Workers
usually cooperate through a load balancing strategy which performs the dynamic
transfer of subspaces to idle workers, and through the exchange of
conflict-clauses~\cite{gradsat,pminisat}.

In 2008, Hamadi et al.~\cite{ManySAT08,ManySATJSAT09} introduced the parallel
portfolio approach. This method exploits the complementarity between different
sequential DPLL strategies to let them compete and cooperate on the original
formula. Since each worker deals with the whole formula, there is no need for
load balancing, and the cooperation is only achieved through the exchange of
learnt clauses. Moreover, the search process is not artificially influenced by
the original set of guiding-path constraints like in the first category of
techniques. With this approach, the crafting of the strategies is important,
and the objective is to cover the space of the search strategies in the best
possible way. 

The main drawback of parallel SAT techniques comes from their required
replication of the formula. This is obvious for the parallel portfolio
approach. It is also true for divide-and-conquer algorithms whose guiding-path
constraints do not produce significantly smaller subproblems (only $log_2 c$
variables have to be set to obtain $c$ subproblems). This makes these
techniques only applicable to problems which fit into the memory of a single
machine. 

In the last two years, portfolio-based parallel solvers became prominent and
it has been used in SMT decision procedures as well~\cite{WintersteigerHM09}. 
We are not aware of a recently developed improvements on the divide-and-conquer
approach (the latest being \cite{pminisat}).
We give a brief description of the parallel solvers qualified for the 2010 SAT
Race\footnote{\url{http://baldur.iti.uka.de/sat-race-2010}}:
 
\begin{itemize}
\item In \texttt{plingeling}~\cite{plingeling}, the original SAT instance is
  duplicated by a boss thread and allocated to worker threads. The strategies
  used by these workers are mainly differentiated around the amount of
  pre-processing, random seeds, and variables branching. Conflict clause
  sharing is restricted to units which are exchanged through the boss
  thread. This solver won the parallel track of the 2010 SAT Race.

\item \texttt{ManySAT} \cite{ManySATJSAT09} was the first parallel SAT
  portfolio. It duplicates the instance of the SAT problem to solve, and runs
  independent SAT solvers differentiated on their restart policies, branching
  heuristics, random seeds, conflict clause learning, etc. It exchanges clauses
  through various policies. Two versions of this solver were presented at the
  2010 SAT Race, they finished second and third.

\item In \texttt{SArTagnan}, \cite{sartagnan} different SAT algorithms are
  allocated to different threads, and differentiated with respect to restart
  policies and VSIDS heuristics. Some threads apply a dynamic resolution
  process~\cite{plingeling,lazyhyperbinary} or exploit reference
  points~\cite{Kottler10}. Some others try to simplify a shared
  clauses database by performing dynamic variable elimination or
  replacement. This solver finished fourth.

\item In \texttt{Antom}~\cite{antom}, the SAT algorithms are differentiated on
  decision heuristic, restart strategy, conflict clause detection, lazy hyper
  binary resolution \cite{plingeling,lazyhyperbinary}, and dynamic unit
  propagation lookahead. Conflict clause sharing is implemented. This solver
  finished fifth.  
\end{itemize}

\subsection{Distributed SAT Solving}

Contrary to parallel SAT, in distributed SAT, the goal is to handle problems
which are by nature distributed or, even more interestingly, to handle problems
which are too large to fit into the memory of a single computing
node. Therefore, the speed-up against a sequential execution is not necessarily
the main objective, and in some cases (large instances) cannot even be measured. 

To the best of our knowledge, the only relevant work in the area presents an
architecture tailored for large distributed Bounded Model
Checking~\cite{GanaiGYA06}. The objective is to perform deep BMC unwindings
thanks to a network of standard computers, where the SAT formulas become so
large that they cannot be handled by any one of the machines. This approach
uses a master/slaves topology, and the unrolling of an instance 
provides a natural partitioning of the problem in a chain of workers. Each
worker has to reconcile its local solution with its neighbors. The master
distributes the parts, and controls the search. First, based on proposals
coming from the slaves, it selects a globally best variable to branch on. From
that decision, each worker performs Boolean Constraint Propagation (BCP) on its
subproblem, and the master performs the same on the globally learnt
clauses. The master maintains the global assignment, and to ensure the
consistency of the parallel BCP algorithms propagates to the slaves Boolean
implications. The master also records the causality of these implications which
allows him to perform conflict-analysis when required.

\subsection{Interpolation}

McMillan's propositional interpolation system~\cite{McMillan03} when employed
in a suitable Model Checking algorithm, has been shown to perform competitively
with algorithms based purely on SAT-solving, i.e., McMillan showed that the
abstraction obtained through interpolation for Model Checking problems is at
least as good as and sometimes better than previously known abstraction
methods.


\section{Lazy Decomposition}

When considering distributed decision procedures, it is usually assumed that
the formulas which are to be solved are too large to be solved on a single
computing node. Under this premise, strategies for distributing a formula have
to be employed. If there exists a quantifier elimination algorithm for the
fragment considered, then it is straight-forward, but comparatively expensive
to distribute the problem: Find sparsely connected partitions of the formula
and eliminate the connections such that the partitions become
independent. For example, let formula $\phi=\phi_1\wedge\phi_2$ where the
partitions $\phi_1$ and $\phi_2$ overlap on variables $X=V_{\phi_1}\cap
V_{\phi_2}$. The elimination of $X$ from $\exists X \;.\; \phi_1\wedge\phi_2$
produces two independent parts $\phi_1'$ and $\phi_2'$, which, respectively,
depend on variables $V_{\phi_1}\setminus X$ and $V_{\phi_2}\setminus X$ and
therefore can be solved independently. While this distribution strategy is
quite simple, it depends on the existence of a quantifier elimination
algorithm. Furthermore, the performance of such an algorithm in practice
depends greatly on the fact that the problem is sparsely connected, which is
not generally a given. We therefore use a different and cheaper method for
distribution: 

\begin{definition}[Lazy Decomposition]
Let $\phi$ be in conjunctive normal form, i.e.,
$\phi=\phi_1\wedge\dots\wedge\phi_n$. A \emph{lazy decomposition} of $\phi$
into $k$ partitions is an equivalent set of formulas $\{\psi_1,\dots,\psi_k\}$
such that each $\psi_i$ is equivalent to some conjunction of clauses from
$\phi$, i.e., there exist $a,b$ $(a<b<n),$ such that
$\psi_i=\phi_a\wedge\dots\wedge\phi_b$.
\end{definition}

We call this a \emph{lazy} decomposition, because no effort is made to ensure
that partitions do not share variables. The formulas $\psi_1\dots\psi_k$ may
then be solved independently, but if the partitions happen to share variables,
i.e., when $V_{\psi_i}\cap V_{\psi_j}\neq\emptyset$ for some $j\neq i$, then
these (potentially global) solutions have to be reconciled. 

Let $S_{\psi_i}=\bigvee_j m_{i,j}$ be the set of all models
satisfying $\psi_i$ and let $S=\{ S_{\psi_1}, \dots, S_{\psi_k} \}$ be a set of
all models of all partitions. The \emph{reconciliation problem} is to
determine whether $S_{\psi_1} \wedge\dots\wedge S_{\psi_k}$ is satisfiable,
i.e., to determine whether there is a global model in $S_\phi$ which has a matching
extension in each $S_{\psi_i}$. Clearly, any set of models is not
required to be any smaller in representation than its corresponding partition;
in fact, it may be exponentially larger. In practice it may
therefore be more efficient to build the solution sets incrementally, avoiding
any blowup wherever possible. To this end, we require the following lemma:
\begin{lemma}\label{lem:decomposition}
Let $\phi=\psi_1\wedge\dots\wedge\psi_k$ and let $m$ be a model for
the shared variables $V:=\bigcup_{i,j=1}^k V_{\psi_i} \cap V_{\psi_j}$ and let
$1\leq i\leq k$. If $I$ is an interpolant for $\neg(\psi_i \wedge m)$ then 
$\phi\Rightarrow I$.
\end{lemma}

\begin{proof}
$I$ is an interpolant for $\neg(\psi_i \wedge m)$ or, equivalently, for
$\psi_i\Rightarrow \neg m$. Therefore $\psi_i\Rightarrow I$. Since
$\phi\Rightarrow\psi_i,$ we also have $\phi\Rightarrow I$.
\end{proof}

\begin{algorithm}[t]
\Input{Formula $\phi$}
\Output{$\top$ if $\phi$ is satisfiable, $\bot$ otherwise}
$\psi_1,\dots,\psi_k$ := $decompose(\phi)$\;
$G$ := $\top$\;
$flag$ := true\;

\While{$flag$}
{
\eIf{$G \equiv \bot$}{\Return $\bot$\;}
{Let $m$ be a total model for $G$\;}

$flag$ := false\;
\ForEach{$i$ in $1\dots k$}
{
  \If{$\psi_i \wedge m \equiv \bot$}
  {
    Let $I$ be an interpolant for $\neg(\psi_i \wedge m)$ over 
    $V_{\psi_i}\cap V_{G}$\; 
    $G$ := $G\wedge I$\;
    $flag$ := true\;
  }
}
}
\Return{$\top$}\;
\caption{\label{alg:recon}A reconciliation algorithm.}
\end{algorithm}

Algorithm~\ref{alg:recon} presents a simple method that makes use of this Lemma
to solve a decomposed formula. First, it extracts a model $m$ for $G$, which is
over the shared variables of the decomposition (the globals). It then attempts to 
extend the model to models satisfying each of the partitions and returns $\top$
if this was successful. Otherwise, it extracts an interpolant $I$ from every
unsatisfiable partition which is subsequently used to refine $G$. When $G$ is
found to be unsatisfiable the algorithm returns $\bot$ as there cannot be any
model that is extensible to all partitions.

The maximum number of iterations require by Algorithm~\ref{alg:recon} depends
on the number and the domain of the shared variables in the decomposition. Of
course, this motivates the use of decomposition techniques that find partitions
with little overlap and on the other hand, motivates the use of interpolation
techniques that produce (logically) weak interpolants such that every
interpolant covers as many (global) models as possible. 

\begin{theorem}
Algorithm~\ref{alg:recon} is sound.
\end{theorem}

\begin{proof}
When the algorithm returns $\top$, there is a model $m$ for $G$ which has an
extension in every $\psi_i$. Conversely, if the algorithm returns $\bot$, then
every potential model $m$ is contained in some interpolant which implies $\neg
m$, which is an immediate consequence of Lemma~\ref{lem:decomposition}.
\end{proof}

\begin{theorem}
Algorithm~\ref{alg:recon} is complete for formulas over finite-domain variables.
\end{theorem}

\begin{proof}
Every iteration of the algorithm excludes at least one possible model from
$G$ (otherwise the algorithm would terminate and return $\top$). For formulas
over finite-domain variables there is only a finite number of potential
variables. Therefore, $G$ has to become unsatisfiable at some point, forcing
$G\equiv\bot$ and therefore termination of the algorithm.
\end{proof}

\section{Interpolation and Conflict Clauses}

The DPLL procedure is an algorithm that solves the SAT problem (see,
e.g.,~\cite{NieuwenhuisOT06}). It does so by evaluating a series of partial
assignments until a total assignment is found. When a partial assignment is
found to be inconsistent with the input formula, DPLL backtracks to a previous
(smaller) assignment. Modern incarnations of this algorithm use
conflict-driven backtracking, which means that the conflicting state of the
solver is analyzed and a \emph{conflict clause} is derived. It is required that
every conflict clause be implied by the original formula, that it is over the
variables of the current assignment, and that it be inconsistent with the
current assignment. Any conflict clause is therefore redundant, but it may help
to prevent further conflicts when it is kept in the clause database (in which
case it is called a \emph{learnt} clause). We think it worthwile to
characterize the relationship between conflict clauses and interpolants:

\begin{corollary}
Every conflict clause for a propositional formula $\phi$ derived under
the partial assignment $\alpha$ is an interpolant for 
$\phi\Rightarrow\neg\alpha.$
\end{corollary}

\begin{proof}
According to the definition of a conflict clause $C$, it must be implied by
$\phi$ and inconsistent with $\alpha$. We therefore have 
$\phi\Rightarrow C$ and
$\neg (C \wedge \alpha)\equiv C\Rightarrow \neg \alpha,$
which makes $C$ an interpolant for $\phi\Rightarrow\neg\alpha$ by
Theorem~\ref{thm:craig}.
\end{proof}

Currently, the most popular conflict resolution scheme for DPLL-style
solvers is the so-called First-UIP method (for a definition
see~\cite{sat-handbook}). The corollary stated above raises the question
whether other interpolation methods are able to improve upon this
scheme. Note that the First-UIP scheme has some properties which make it very
efficient in practice: 
\begin{itemize}
\item a conflict clause can be computed in linear time and
\item every such clause is \emph{asserting}, i.e., it contains a unique literal
  which is unassigned after backtracking.
\end{itemize}

More general interpolation schemes like McMillan's or HKP also have the first
of these properties since interpolants are usually computed in linear time from a
resolution proof. Therefore, an interpolant may be computed in linear time,
too, if the resolution proof of the current conflict is kept in the state of
the solver. This, however, is much more expensive than keeping the reasons for
implications as is done in the First-UIP scheme.

An interpolant is generally not of clause form. If it is to be kept as part of
the problem (akin to a learnt clause) it therefore requires conversion. The
straight-forward expansion to CNF may increase the size of the interpolant
exponentially. The Tseitin transformation~\cite{tseitin68} increases the size
of the interpolant only linearly, but introduces new variables. It is not clear
which of these methods is to be preferred. In general, however, a \emph{set} of
conflict clauses is produced, instead of a single clause like in the First-UIP
scheme.

An interpolant is also asserting in the sense that it is asserted to be true;
however, it is not immediately asserting a specific literal like a First-UIP
conflict clause. A preliminary experimental evaluation (of which the details
are omitted) has shown that none of the known propositional interpolation
methods performs better than the First-UIP scheme. This, however, may be due to
the lack of an efficient interpolation algorithm that matches the performance
of the algorithm for First-UIP conflict resolution. 


\section{Experimental Evaluation}
\label{sec:experiments}

As a first step in evaluating our algorithm, we implemented a propositional
satisfiability solver based on the MiniSAT solver. We restrict ourselves to
the slightly outdated version 1.14p, because propositional interpolation
methods require proof production, which is not available in more recent
versions of MiniSAT~\cite{EenS03}. Interpolants are produced by iterating over
the resolution proof, which is saved (explicitly) in memory. We use Reduced
Boolean Circuits (RBCs~\cite{AbdullaBE00}) to represent interpolants such that
recurring structure is exploited. Furthermore, Algorithm~\ref{alg:recon}
permits the exploitation of state-of-the-art SAT solver technology, like
incremental solving techniques in solving partitions. Furthermore, every
assignment to the globals is a set of clauses of size 1, which means that
facilities for solving a formula under assumptions may be made use of. 
The lazy decomposition used by our implementation is indeed quite trivial: it
simply divides the clauses of the problem into a predefined number $p$ of
equally sized partitions. Clauses are ordered as they appear in the input file
and each partition $i$ is assigned the clauses numbered from $i \cdot
\frac{n}{p}$ to $(i+1) \cdot \frac{n}{p}$, where $n$ is the total number of
clauses.

Our implementation is evaluated on set of formulas which are small but hard to
solve~\cite{AloulRMS02}\footnote{\url{http://www.aloul.net/benchmarks.html}}. They
are known to contain symmetries, which potentially can be exploited by
interpolation. For this evaluation, our implementation uses only a single
processing element, i.e., the evaluation of the partitions of a decomposition
is sequentialized. Through this, we are able to show that our algorithm
performs well even when using the same resources as a traditional
solver. Preliminary experiments have shown that an actual (shared-memory)
parallelization of our algorithm performs better than the sequentialized
version, but not significantly so, which is due to the lack of a load-balancing
mechanism to balance the runtime of partition evaluation.

All our experiments are executed on a Windows HPC cluster of dual Quad-Xeon
2.5~GHz processors with 16~GB of memory each, using a timeout of 3600~seconds
and a memory limit of 2~GB.

\begin{figure}[p]
\centering
\subfloat[][\label{fig:symmetry-m}McMillan interpolants]
{ 
  \input{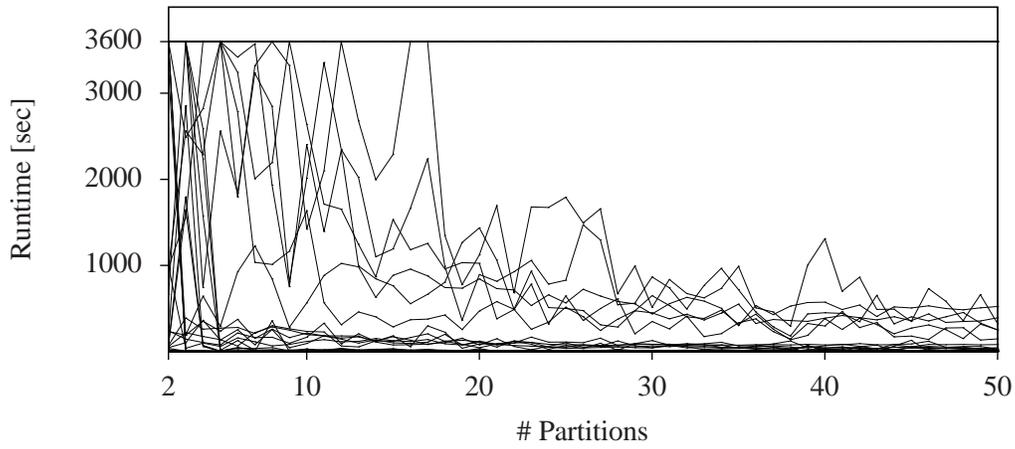}
}
\quad
\subfloat[][\label{fig:symmetry-p}HKP interpolants]
{
  \input{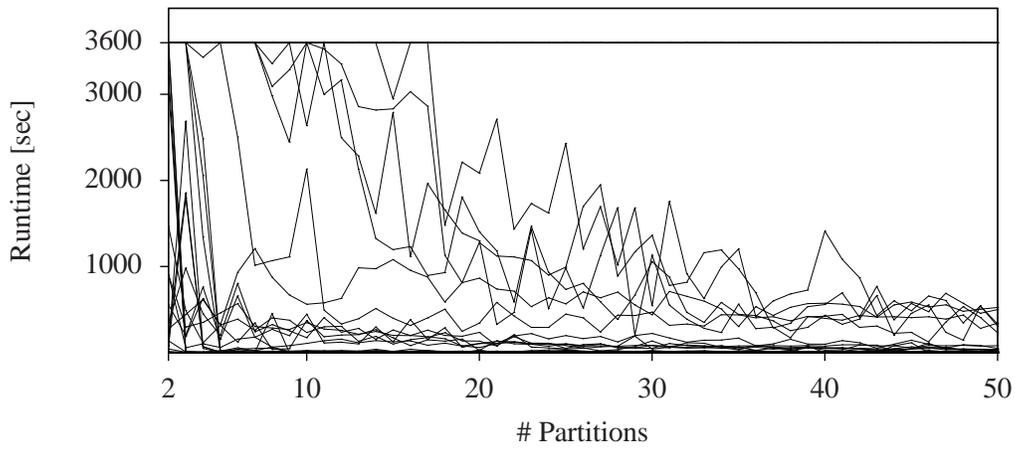}
}
\qquad
\subfloat[][\label{fig:symmetry-im}Dual McMillan interpolants]
{
  \input{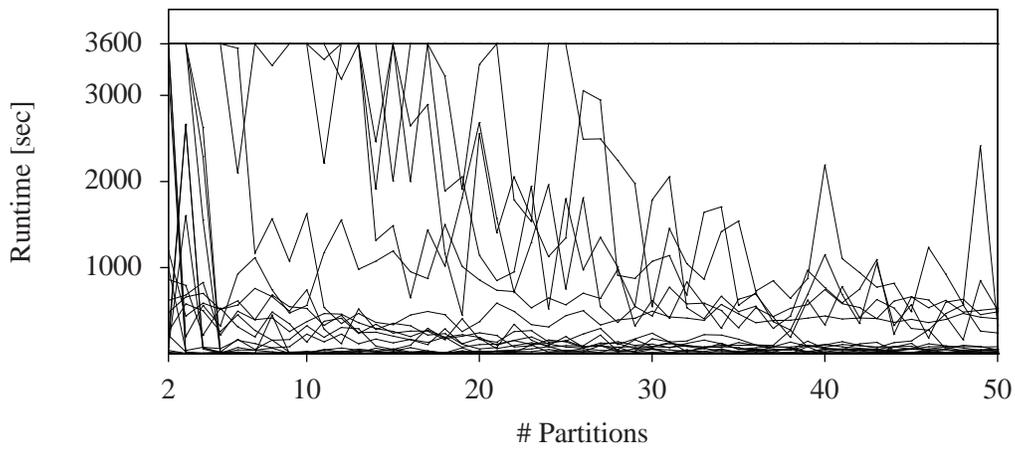}
}
\caption{Decomposition into $\{2,\dots,50\}$ partitions, using different
  interpolation systems.}
\end{figure}

\begin{table}[p]
\centering
\scriptsize

\begin{tabular}{|l||r|r||r|r|r|r|r|r||r|}\hline
&\multicolumn{2}{|c||}{\bf Minisat}&\multicolumn{7}{|c|}{\bf Decomposition (\# partitions)}\\\hline
\bf Filename & \bf 2.2.0 & \bf 1.14p & \bf 5 & \bf 10 & \bf 20 & \bf 30 & \bf 40 & \bf 50 & \bf Best\\\hline\hline
chnl10\_11.cnf & 169.39 & 33.52 & 3.00 & 3.20 & 1.92 & 2.14 & 1.29 & \bf 1.22 & \bf 0.98\\\hline
chnl10\_12.cnf & 165.02 & 20.95 & 2.46 & 2.07 & 4.17 & 2.07 & \bf 1.53 & 2.15 & \bf 0.98\\\hline
chnl10\_13.cnf & 204.77 & 15.23 & 2.93 & \bf 1.87 & 3.40 & 2.57 & 2.40 & 2.07 & \bf 0.89\\\hline
chnl11\_12.cnf & T/O & 291.61 & \bf 5.23 & 6.05 & 17.02 & 14.29 & 6.46 & 7.16 & \bf 4.45\\\hline
chnl11\_13.cnf & T/O & 960.72 & 5.69 & \bf 4.70 & 18.28 & 12.62 & 7.55 & 6.07 & \bf 4.45\\\hline
chnl11\_20.cnf & T/O & 2346.30 & 13.31 & 15.07 & \bf 5.05 & 10.50 & 10.33 & 12.39 & \bf 4.96\\\hline
fpga10\_8\_sat.cnf & \bf 0.00 & 0.02 & 0.02 & 0.02 & 0.03 & 0.03 & 0.03 & 0.06 & \bf 0.00\\\hline
fpga10\_9\_sat.cnf & \bf 0.02 & \bf 0.02 & \bf 0.02 & \bf 0.02 & \bf 0.02 &
0.03 & 0.03 & 0.03 & \bf 0.02\\\hline
fpga12\_8\_sat.cnf & \bf 0.00 & 0.02 & 0.02 & 0.02 & 0.03 & 0.03 & 0.03 & 0.05 & \bf 0.00\\\hline
fpga12\_9\_sat.cnf & 0.03 & \bf 0.02 & 0.03 & \bf 0.02 & 0.05 & 0.05 & 0.05 & 0.05 & \bf 0.02\\\hline
fpga12\_11\_sat.cnf & 0.02 & \bf 0.00 & 0.03 & 0.03 & 0.03 & 0.03 & 0.05 & 0.06 & \bf 0.00\\\hline
fpga12\_12\_sat.cnf & \bf 0.00 & 0.02 & 0.05 & 0.05 & 0.03 & 0.03 & 0.03 & 0.05 & 0.02\\\hline
fpga13\_9\_sat.cnf & \bf 0.00 & \bf 0.00 & 0.03 & 0.02 & 0.03 & 0.06 & 0.05 &
0.08 & 0.02\\\hline
fpga13\_10\_sat.cnf & \bf 0.00 & \bf 0.00 & 0.05 & 0.03 & 0.05 & 0.05 & 0.05 & 0.06 & 0.02\\\hline
fpga13\_12\_sat.cnf & \bf 0.00 & \bf 0.00 & 0.06 & 0.08 & 0.11 & 0.09 & 0.12 & 0.12 & 0.02\\\hline
hole7.cnf & 0.08 & \bf 0.05 & 0.08 & 0.11 & 0.08 & 0.08 & 0.11 & 0.11 & 0.06\\\hline
hole8.cnf & 0.37 & 0.33 & \bf 0.31 & 0.51 & 0.37 & 0.53 & 0.37 & 0.47 & \bf 0.30\\\hline
hole9.cnf & 6.51 & 3.59 & 1.56 & 3.03 & 2.31 & 1.93 & \bf 1.40 & 1.65 & \bf 1.22\\\hline
hole10.cnf & 247.11 & 26.75 & \bf 10.09 & 22.76 & 11.22 & 12.93 & 12.50 & 10.65 & \bf 4.34\\\hline
hole11.cnf & T/O & 509.86 & \bf 59.12 & 116.42 & 109.61 & 81.88 & 75.93 & 80.84 & \bf 37.03\\\hline
hole12.cnf & T/O & T/O & \bf 293.08 & 564.38 & 843.20 & 435.18 & 572.90 & 522.65 & \bf 187.70\\\hline
s3-3-3-1.cnf & 0.44 & \bf 0.16 & 2.98 & 0.30 & 0.20 & 0.23 & 0.27 & 0.39 & \bf 0.16\\\hline
s3-3-3-3.cnf & 1.20 & \bf 0.02 & 1.90 & 0.33 & 0.23 & 0.28 & 0.50 & 0.39 & 0.12\\\hline
s3-3-3-4.cnf & 0.44 & 1.50 & 3.51 & 0.23 & \bf 0.14 & 0.19 & 0.42 & 0.51 & \bf 0.12\\\hline
s3-3-3-8.cnf & 0.34 & 0.80 & 0.98 & 0.37 & \bf 0.20 & 0.48 & 0.56 & 0.42 & \bf
0.19\\\hline
s3-3-3-10.cnf & 0.47 & 0.81 & 1.23 & 0.41 & \bf 0.36 & 0.48 & 0.67 & 0.55 & \bf 0.28\\\hline
Urq3\_5.cnf & 310.22 & \bf 58.27 & T/O & 1633.50 & 470.33 & 448.89 & 386.04 & 391.56 & 243.55\\\hline
Urq4\_5.cnf & T/O & T/O & T/O & T/O & T/O & T/O & T/O & T/O & T/O\\\hline
Urq5\_5.cnf & T/O & T/O & T/O & T/O & T/O & T/O & T/O & T/O & T/O\\\hline
Urq6\_5.cnf & T/O & T/O & T/O & T/O & T/O & T/O & T/O & T/O & T/O\\\hline
Urq7\_5.cnf & T/O & T/O & M/O & T/O & T/O & T/O & T/O & T/O & T/O\\\hline
Urq8\_5.cnf & T/O & T/O & T/O & T/O & T/O & T/O & T/O & T/O & T/O\\\hline
fpga10\_8\_sat\_rcr.cnf & \bf 0.00 & \bf 0.00 & 0.03 & 0.03 & 0.03 & 0.03 & 0.05 & 0.05 & 0.02\\\hline
fpga10\_9\_sat\_rcr.cnf & \bf 0.00 & \bf 0.00 & 0.02 & 0.05 & 0.03 & 0.05 & 0.06 & 0.08 & \bf 0.00\\\hline
fpga12\_8\_sat\_rcr.cnf & \bf 0.00 & 0.02 & 0.03 & 0.02 & 0.08 & 0.03 & 0.09 & 0.06 & 0.02\\\hline
fpga12\_9\_sat\_rcr.cnf & \bf 0.02 & \bf 0.02 & \bf 0.02 & 0.03 & 0.05 & 0.05 & 0.05 & 0.08 & \bf 0.02\\\hline
fpga12\_10\_sat\_rcr.cnf & 0.02 & \bf 0.00 & 0.03 & 0.05 & 0.05 & 0.08 & 0.09 & 0.08 & 0.02\\\hline
fpga12\_11\_sat\_rcr.cnf & \bf 0.02 & 0.03 & \bf 0.02 & 0.06 & 0.06 & 0.09 & 0.08 & 0.09 & \bf 0.02\\\hline
fpga12\_12\_sat\_rcr.cnf & 0.02 & \bf 0.00 & 0.05 & 0.50 & 0.42 & 0.09 & 0.14 & 0.19 & 0.02\\\hline
fpga13\_9\_sat\_rcr.cnf & \bf 0.00 & 0.02 & 0.03 & 0.03 & 0.05 & 0.05 & 0.06 & 0.09 & 0.02\\\hline
fpga13\_10\_sat\_rcr.cnf & \bf 0.03 & 0.05 & \bf 0.03 & 0.05 & 0.14 & 0.08 & 0.08 & 0.14 & \bf 0.02\\\hline
fpga13\_11\_sat\_rcr.cnf & \bf 0.00 & \bf 0.00 & 0.05 & 0.03 & 0.06 & 0.09 & 0.12 & 0.11 & 0.03\\\hline
fpga13\_12\_sat\_rcr.cnf & \bf 0.00 & \bf 0.00 & 0.08 & 30.83 & 34.30 & 0.12 & 0.25 & 0.12 & 0.02\\\hline
fpga10\_11\_uns\_rcr.cnf & 173.04 & 32.09 & 235.20 & 213.47 & 110.25 & 41.68 & 41.23 & \bf 29.39 & \bf 26.97\\\hline
fpga10\_12\_uns\_rcr.cnf & 484.76 & 109.31 & 137.37 & 197.23 & 120.71 & 41.65 & 25.79 & \bf 13.81 & \bf 13.81\\\hline
fpga10\_13\_uns\_rcr.cnf & 1279.94 & 110.32 & 132.26 & 92.70 & 74.96 & 65.26 & \bf 38.77 & 42.46 & \bf 20.12\\\hline
fpga10\_15\_uns\_rcr.cnf & T/O & 215.50 & 271.72 & 232.78 & 42.85 & 83.32 & 26.13 & \bf 20.64 & \bf 8.25\\\hline
fpga10\_20\_uns\_rcr.cnf & 3089.90 & 696.29 & 76.05 & 164.60 & 87.72 & 125.91 & 65.74 & \bf 28.24 & \bf 28.24\\\hline
fpga11\_12\_uns\_rcr.cnf & 3204.84 & 449.22 & T/O & 2404.97 & 1435.07 & 648.40 & 380.16 & \bf 251.60 & \bf 134.86\\\hline
fpga11\_13\_uns\_rcr.cnf & T/O & 1933.98 & T/O & 2638.34 & 1023.83 & 348.48 & 295.11 & \bf 146.20 & \bf 132.34\\\hline
fpga11\_14\_uns\_rcr.cnf & T/O & 3199.81 & 2562.86 & 2011.51 & 893.40 & 866.23 & 450.22 & \bf 250.13 & \bf 179.15\\\hline
fpga11\_15\_uns\_rcr.cnf & T/O & T/O & T/O & 1417.94 & 1121.73 & 512.84 & 1309.82 & \bf 321.85 & \bf 294.67\\\hline
fpga11\_20\_uns\_rcr.cnf & T/O & T/O & T/O & 2828.49 & 2141.49 & 1068.31 & 577.14 & \bf 352.92 & \bf 154.24\\\hline
\end{tabular}

\caption{\label{tbl:runtimes}Runtime comparison with MiniSAT (versions 2.2.0
  and 1.14p), using McMillan's interpolation system. Bold numbers indicate
  the smallest runtime in the decompositions shown here or the best
  decomposition into $\{2,\dots,50\}$ partitions (right-most column).}
\end{table}

To assess the impact of the decomposition on the solver performance, we
investigate all decompositions into 2 to 50 partitions for each of the three
interpolation methods (McMillan's, Dual McMillan's and
HKP). Each line in Figures~\ref{fig:symmetry-m},~\ref{fig:symmetry-p},
and~\ref{fig:symmetry-im} corresponds to one benchmark file and indicates the
change in runtime required to solve the benchmark when using from 2 up to 50
partitions in the decomposition. These figures provide strong evidence for an
improvement of the runtime behavior as the number of partitions increase. Note
that, as mentioned before, the evaluation of partitions was sequentialized for
this experiment, i.e., this effect is \emph{not} due to an increasing number of
resources being utilized.
The graphs in Figures~\ref{fig:symmetry-m},~\ref{fig:symmetry-p},
and~\ref{fig:symmetry-im} provide equally strong evidence for the utility of
McMillan's interpolation system: the impact on the runtime is the largest and
most consistent of all three interpolation systems.

Finally, Table~\ref{tbl:runtimes} provides a comparison of the runtime of
MiniSAT versions 2.2.0 and 1.14p with a selection of different decompositions,
the right-most column indicating the time of the best decomposition found
among all those evaluated. It is clear from this table that no single
partitioning can be identified as the best overall method. However, some
decompositions, like the one into 50 partitions, perform consistently well and
almost always better than either versions of MiniSAT.


\section{Conclusion}

We present the concept of lazy distribution for first-order decision
procedures. Formulas are decomposed into partitions without the need for
quantifier elimination or any other method for logical disconnection of the
partitions. Instead, local models for the partitions are reconciled globally
through the use of Craig interpolation. Experiments using different
interpolation systems and decompositions for propositional formulas indicate
that our approach performs better than traditional solving methods even when
sequentialized, i.e., when no additional resources are used. At the same time,
our algorithm provides straight-forward opportunities for parallelization and
distribution of the solving process.

\nocite{*}
\bibliographystyle{eptcs}
\bibliography{literature}

\begin{thebibliography}{10}
\providecommand{\bibitemdeclare}[2]{}
\providecommand{\urlprefix}{Available at }
\providecommand{\url}[1]{\texttt{#1}}
\providecommand{\href}[2]{\texttt{#2}}
\providecommand{\urlalt}[2]{\href{#1}{#2}}
\providecommand{\doi}[1]{doi:\urlalt{http://dx.doi.org/#1}{#1}}
\providecommand{\bibinfo}[2]{#2}

\bibitemdeclare{inproceedings}{AbdullaBE00}
\bibitem{AbdullaBE00}
\bibinfo{author}{Parosh~Aziz Abdulla}, \bibinfo{author}{Per Bjesse} \&
  \bibinfo{author}{Niklas E{\'e}n} (\bibinfo{year}{2000}):
  \emph{\bibinfo{title}{Symbolic Reachability Analysis Based on SAT-Solvers}}.
\newblock In: {\sl \bibinfo{booktitle}{Proc. of TACAS}}, {\sl
  \bibinfo{series}{LNCS}} \bibinfo{volume}{1785},
  \bibinfo{publisher}{Springer}, pp. \bibinfo{pages}{411--425},
  \doi{10.1007/3-540-46419-0_28}.

\bibitemdeclare{inproceedings}{AloulRMS02}
\bibitem{AloulRMS02}
\bibinfo{author}{Fadi~A. Aloul}, \bibinfo{author}{Arathi Ramani},
  \bibinfo{author}{Igor~L. Markov} \& \bibinfo{author}{Karem~A. Sakallah}
  (\bibinfo{year}{2002}): \emph{\bibinfo{title}{Solving difficult {SAT}
  instances in the presence of symmetry}}.
\newblock In: {\sl \bibinfo{booktitle}{Proc. of DAC}},
  \bibinfo{publisher}{ACM}, pp. \bibinfo{pages}{731--736},
  \doi{10.1145/513918.514102}.

\bibitemdeclare{inproceedings}{AlpernHRSZ90}
\bibitem{AlpernHRSZ90}
\bibinfo{author}{Bowen Alpern}, \bibinfo{author}{Roger Hoover},
  \bibinfo{author}{Barry~K. Rosen}, \bibinfo{author}{Peter~F. Sweeney} \&
  \bibinfo{author}{F.~Kenneth Zadeck} (\bibinfo{year}{1990}):
  \emph{\bibinfo{title}{Incremental Evaluation of Computational Circuits}}.
\newblock In: {\sl \bibinfo{booktitle}{Proc. of the ACM-SIAM Symposium on
  Discrete Algorithms (SODA)}}, \bibinfo{publisher}{ACM}, pp.
  \bibinfo{pages}{32--42}.

\bibitemdeclare{article}{AspvallPT79}
\bibitem{AspvallPT79}
\bibinfo{author}{Bengt Aspvall}, \bibinfo{author}{Michael~F. Plass} \&
  \bibinfo{author}{Robert~Endre Tarjan} (\bibinfo{year}{1979}):
  \emph{\bibinfo{title}{A Linear-Time Algorithm for Testing the Truth of
  Certain Quantified {Boolean} Formulas}}.
\newblock {\sl \bibinfo{journal}{Inf. Process. Lett.}}
  \bibinfo{volume}{8}(\bibinfo{number}{3}), pp. \bibinfo{pages}{121--123},
  \doi{10.1016/0020-0190(79)90002-4}.

\bibitemdeclare{techreport}{lazyhyperbinary}
\bibitem{lazyhyperbinary}
\bibinfo{author}{Armin Biere} (\bibinfo{year}{2009}):
  \emph{\bibinfo{title}{Lazy hyper binary resolution}}.
\newblock \bibinfo{type}{Technical Report}, \bibinfo{institution}{Dagstuhl
  Seminar 09461}.

\bibitemdeclare{techreport}{plingeling}
\bibitem{plingeling}
\bibinfo{author}{Armin Biere} (\bibinfo{year}{2010}):
  \emph{\bibinfo{title}{Lingeling, Plingeling, PicoSAT and PrecoSAT at {SAT}
  Race 2010}}.
\newblock \bibinfo{type}{Technical Report} \bibinfo{number}{10/1},
  \bibinfo{institution}{FMV Reports Series}.

\bibitemdeclare{proceedings}{sat-handbook}
\bibitem{sat-handbook}
\bibinfo{editor}{Armin Biere}, \bibinfo{editor}{Marijn Heule},
  \bibinfo{editor}{Hans van Maaren} \& \bibinfo{editor}{Toby Walsh}, editors
  (\bibinfo{year}{2009}): \emph{\bibinfo{title}{Handbook of Satisfiability}}.
  {\sl \bibinfo{series}{Frontiers in Artificial Intelligence and Applications}}
  \bibinfo{volume}{185}, \bibinfo{publisher}{IOS Press}.

\bibitemdeclare{inproceedings}{BjesseKDSZ03}
\bibitem{BjesseKDSZ03}
\bibinfo{author}{Per Bjesse}, \bibinfo{author}{James~H. Kukula},
  \bibinfo{author}{Robert~F. Damiano}, \bibinfo{author}{Ted Stanion} \&
  \bibinfo{author}{Yunshan Zhu} (\bibinfo{year}{2004}):
  \emph{\bibinfo{title}{Guiding SAT Diagnosis with Tree Decompositions}}.
\newblock In: {\sl \bibinfo{booktitle}{Theory and Applications of
  Satisfiability Testing, 6th International Conference, SAT 2003, Selected
  Revised Papers}}, {\sl \bibinfo{series}{Lecture Notes in Computer Science}}
  \bibinfo{volume}{2919}, \bibinfo{publisher}{Springer}, pp.
  \bibinfo{pages}{315--329}, \doi{10.1007/978-3-540-24605-3_24}.

\bibitemdeclare{inproceedings}{CamposNZS09}
\bibitem{CamposNZS09}
\bibinfo{author}{S.~Campos}, \bibinfo{author}{J.~Neves},
  \bibinfo{author}{L.~Zarate} \& \bibinfo{author}{M.~Song}
  (\bibinfo{year}{2009}): \emph{\bibinfo{title}{Distributed {BMC}: A
  Depth-First Approach to Explore Clause Symmetry}}.
\newblock In: {\sl \bibinfo{booktitle}{Proc. of the Intl. Conf. and Workshop on
  the Engineering of Computer Based Systems (ECBS 2009)}},
  \bibinfo{publisher}{IEEE Press}, pp. \bibinfo{pages}{89--94},
  \doi{10.1109/ECBS.2009.26}.

\bibitemdeclare{techreport}{gradsat}
\bibitem{gradsat}
\bibinfo{author}{Wahid Chrabakh} \& \bibinfo{author}{Rich Wolski}
  (\bibinfo{year}{2003}): \emph{\bibinfo{title}{{G}r{ADSAT}: A Parallel
  \protect{SAT} Solver for the Grid}}.
\newblock \bibinfo{type}{Technical Report}, \bibinfo{institution}{UCSB Computer
  Science}.

\bibitemdeclare{article}{Craig57}
\bibitem{Craig57}
\bibinfo{author}{William Craig} (\bibinfo{year}{1957}):
  \emph{\bibinfo{title}{Linear Reasoning. A New Form of the
  {Herbrand}-{Gentzen} Theorem}}.
\newblock {\sl \bibinfo{journal}{J. Symb. Log.}}
  \bibinfo{volume}{22}(\bibinfo{number}{3}), pp. \bibinfo{pages}{250--268},
  \doi{10.2307/2963593}.

\bibitemdeclare{article}{DPLL}
\bibitem{DPLL}
\bibinfo{author}{Martin Davis}, \bibinfo{author}{George Logemann} \&
  \bibinfo{author}{Donald Loveland} (\bibinfo{year}{1962}):
  \emph{\bibinfo{title}{A machine program for theorem-proving}}.
\newblock {\sl \bibinfo{journal}{Commun. ACM}} \bibinfo{volume}{5}, pp.
  \bibinfo{pages}{394--397}, \doi{10.1145/368273.368557}.

\bibitemdeclare{inproceedings}{DSilva10}
\bibitem{DSilva10}
\bibinfo{author}{Vijay D'Silva} (\bibinfo{year}{2010}):
  \emph{\bibinfo{title}{Propositional Interpolation and Abstract
  Interpretation}}.
\newblock In: {\sl \bibinfo{booktitle}{Proc. of ESOP}}, {\sl
  \bibinfo{series}{LNCS}} \bibinfo{volume}{6012},
  \bibinfo{publisher}{Springer}, pp. \bibinfo{pages}{185--204},
  \doi{10.1007/978-3-642-11957-6_11}.

\bibitemdeclare{inproceedings}{EenS03}
\bibitem{EenS03}
\bibinfo{author}{Niklas E{\'e}n} \& \bibinfo{author}{Niklas S{\"o}rensson}
  (\bibinfo{year}{2003}): \emph{\bibinfo{title}{An Extensible SAT-solver}}.
\newblock In: {\sl \bibinfo{booktitle}{Proc. of SAT}}, {\sl
  \bibinfo{series}{LNCS}} \bibinfo{volume}{2919},
  \bibinfo{publisher}{Springer}, pp. \bibinfo{pages}{502--518},
  \doi{10.1007/978-3-540-24605-3_37}.

\bibitemdeclare{article}{GanaiGYA06}
\bibitem{GanaiGYA06}
\bibinfo{author}{Malay~K. Ganai}, \bibinfo{author}{Aarti Gupta},
  \bibinfo{author}{Zijiang Yang} \& \bibinfo{author}{Pranav Ashar}
  (\bibinfo{year}{2006}): \emph{\bibinfo{title}{Efficient distributed {SAT} and
  {SAT}-based distributed Bounded Model Checking}}.
\newblock {\sl \bibinfo{journal}{Intl. J. on Software Tools for Technology
  Transfer (STTT)}} \bibinfo{volume}{8}(\bibinfo{number}{4-5}), pp.
  \bibinfo{pages}{387--396}, \doi{10.1007/s10009-005-0203-z}.

\bibitemdeclare{techreport}{pminisat}
\bibitem{pminisat}
\bibinfo{author}{Chu Gorey} \& \bibinfo{author}{Peter~J. Stuckey}
  (\bibinfo{year}{2008}): \emph{\bibinfo{title}{{PMiniSAT}: a parallelization
  of {MiniSAT} 2.0}}.
\newblock \bibinfo{type}{Technical Report}, \bibinfo{institution}{{SAT} Race}.

\bibitemdeclare{inproceedings}{HaeuplerKMST08}
\bibitem{HaeuplerKMST08}
\bibinfo{author}{Bernhard Haeupler}, \bibinfo{author}{Telikepalli Kavitha},
  \bibinfo{author}{Rogers Mathew}, \bibinfo{author}{Siddhartha Sen} \&
  \bibinfo{author}{Robert~Endre Tarjan} (\bibinfo{year}{2008}):
  \emph{\bibinfo{title}{Faster Algorithms for Incremental Topological
  Ordering}}.
\newblock In: {\sl \bibinfo{booktitle}{Proc. of ICALP}}, {\sl
  \bibinfo{series}{LNCS}} \bibinfo{volume}{5125},
  \bibinfo{publisher}{Springer}, pp. \bibinfo{pages}{421--433},
  \doi{10.1007/978-3-540-70575-8_35}.

\bibitemdeclare{article}{ManySAT08}
\bibitem{ManySAT08}
\bibinfo{author}{Youssef Hamadi}, \bibinfo{author}{Said Jabbour} \&
  \bibinfo{author}{Lakhdar Sais} (\bibinfo{year}{2008}):
  \emph{\bibinfo{title}{Many{SAT}: Solver Description}}.
\newblock {\sl \bibinfo{journal}{Microsoft Research Technical Report
  MSR-TR-2008-83}} .

\bibitemdeclare{article}{ManySATJSAT09}
\bibitem{ManySATJSAT09}
\bibinfo{author}{Youssef Hamadi}, \bibinfo{author}{Said Jabbour} \&
  \bibinfo{author}{Lakhdar Sais} (\bibinfo{year}{2009}):
  \emph{\bibinfo{title}{Many{SAT}: a Parallel {SAT} Solver}}.
\newblock {\sl \bibinfo{journal}{J. on Satisfiability, Boolean Modeling and
  Computation}} \bibinfo{volume}{6}, pp. \bibinfo{pages}{245--262}.

\bibitemdeclare{inproceedings}{Huang95}
\bibitem{Huang95}
\bibinfo{author}{Guoxiang Huang} (\bibinfo{year}{1995}):
  \emph{\bibinfo{title}{Constructing Craig Interpolation Formulas}}.
\newblock In: {\sl \bibinfo{booktitle}{Proc. of the Intl. Conf. on Computing
  and Combinatorics (COCOON)}}, {\sl \bibinfo{series}{LNCS}}
  \bibinfo{volume}{959}, \bibinfo{publisher}{Springer}, pp.
  \bibinfo{pages}{181--190}.

\bibitemdeclare{inproceedings}{HyvarinenJN10}
\bibitem{HyvarinenJN10}
\bibinfo{author}{Antti Eero~Johannes Hyv{\"a}rinen}, \bibinfo{author}{Tommi~A.
  Junttila} \& \bibinfo{author}{Ilkka Niemel{\"a}} (\bibinfo{year}{2010}):
  \emph{\bibinfo{title}{Partitioning SAT Instances for Distributed Solving}}.
\newblock In: {\sl \bibinfo{booktitle}{Proc. of LPAR}}, {\sl
  \bibinfo{series}{LNCS}} \bibinfo{volume}{6397},
  \bibinfo{publisher}{Springer}, pp. \bibinfo{pages}{372--386},
  \doi{10.1007/978-3-642-16242-8_27}.

\bibitemdeclare{techreport}{sartagnan}
\bibitem{sartagnan}
\bibinfo{author}{Stephan Kottler} (\bibinfo{year}{2010}):
  \emph{\bibinfo{title}{{SA}r{T}agnan: Solver Description}}.
\newblock \bibinfo{type}{Technical Report}, \bibinfo{institution}{{SAT} Race
  2010}.

\bibitemdeclare{inproceedings}{Kottler10}
\bibitem{Kottler10}
\bibinfo{author}{Stephan Kottler} (\bibinfo{year}{2010}):
  \emph{\bibinfo{title}{{SAT} Solving with Reference Points}}.
\newblock In: {\sl \bibinfo{booktitle}{Proc. of SAT}}, {\sl
  \bibinfo{series}{LNCS}} \bibinfo{volume}{6175},
  \bibinfo{publisher}{Springer}, pp. \bibinfo{pages}{143--157},
  \doi{10.1007/978-3-642-14186-7_13}.

\bibitemdeclare{article}{Krajicek97}
\bibitem{Krajicek97}
\bibinfo{author}{Jan Kraj\'{\i}cek} (\bibinfo{year}{1997}):
  \emph{\bibinfo{title}{Interpolation Theorems, Lower Bounds for Proof Systems,
  and Independence Results for Bounded Arithmetic}}.
\newblock {\sl \bibinfo{journal}{J. Symb. Log.}}
  \bibinfo{volume}{62}(\bibinfo{number}{2}), pp. \bibinfo{pages}{457--486},
  \doi{10.2307/2275541}.

\bibitemdeclare{article}{MarchettiNR96}
\bibitem{MarchettiNR96}
\bibinfo{author}{Alberto Marchetti-Spaccamela}, \bibinfo{author}{Umberto Nanni}
  \& \bibinfo{author}{Hans Rohnert} (\bibinfo{year}{1996}):
  \emph{\bibinfo{title}{Maintaining a Topological Order Under Edge
  Insertions}}.
\newblock {\sl \bibinfo{journal}{Inf. Process. Lett.}}
  \bibinfo{volume}{59}(\bibinfo{number}{1}), pp. \bibinfo{pages}{53--58},
  \doi{10.1016/0020-0190(96)00075-0}.

\bibitemdeclare{inproceedings}{McMillan03}
\bibitem{McMillan03}
\bibinfo{author}{Kenneth~L. McMillan} (\bibinfo{year}{2003}):
  \emph{\bibinfo{title}{Interpolation and {SAT}-Based Model Checking}}.
\newblock In: {\sl \bibinfo{booktitle}{Proc. of CAV}}, {\sl
  \bibinfo{series}{LNCS}} \bibinfo{volume}{2725},
  \bibinfo{publisher}{Springer}, pp. \bibinfo{pages}{1--13}.

\bibitemdeclare{inproceedings}{MoskewiczMZZM01}
\bibitem{MoskewiczMZZM01}
\bibinfo{author}{Matthew~W. Moskewicz}, \bibinfo{author}{Conor~F. Madigan},
  \bibinfo{author}{Ying Zhao}, \bibinfo{author}{Lintao Zhang} \&
  \bibinfo{author}{Sharad Malik} (\bibinfo{year}{2001}):
  \emph{\bibinfo{title}{Chaff: Engineering an Efficient {SAT} Solver}}.
\newblock In: {\sl \bibinfo{booktitle}{Proc. of DAC}},
  \bibinfo{publisher}{ACM}, pp. \bibinfo{pages}{530--535}.

\bibitemdeclare{article}{NieuwenhuisOT06}
\bibitem{NieuwenhuisOT06}
\bibinfo{author}{Robert Nieuwenhuis}, \bibinfo{author}{Albert Oliveras} \&
  \bibinfo{author}{Cesare Tinelli} (\bibinfo{year}{2006}):
  \emph{\bibinfo{title}{Solving {SAT} and {SAT} Modulo Theories: From an
  abstract {Davis}--{Putnam}--{Logemann}--{Loveland} procedure to {DPLL(T)}}}.
\newblock {\sl \bibinfo{journal}{J. ACM}}
  \bibinfo{volume}{53}(\bibinfo{number}{6}), pp. \bibinfo{pages}{937--977},
  \doi{10.1145/1217856.1217859}.

\bibitemdeclare{article}{PearceK06}
\bibitem{PearceK06}
\bibinfo{author}{David~J. Pearce} \& \bibinfo{author}{Paul H.~J. Kelly}
  (\bibinfo{year}{2006}): \emph{\bibinfo{title}{A dynamic topological sort
  algorithm for directed acyclic graphs}}.
\newblock {\sl \bibinfo{journal}{ACM J. of Experimental Algorithmics}}
  \bibinfo{volume}{11}, \doi{10.1145/1187436.1210590}.

\bibitemdeclare{article}{Pudlak97}
\bibitem{Pudlak97}
\bibinfo{author}{Pavel Pudl{\'a}k} (\bibinfo{year}{1997}):
  \emph{\bibinfo{title}{Lower Bounds for Resolution and Cutting Plane Proofs
  and Monotone Computations}}.
\newblock {\sl \bibinfo{journal}{J. Symb. Log.}}
  \bibinfo{volume}{62}(\bibinfo{number}{3}), pp. \bibinfo{pages}{981--998},
  \doi{10.2307/2275583}.

\bibitemdeclare{article}{RodittyZ08}
\bibitem{RodittyZ08}
\bibinfo{author}{Liam Roditty} \& \bibinfo{author}{Uri Zwick}
  (\bibinfo{year}{2008}): \emph{\bibinfo{title}{Improved Dynamic Reachability
  Algorithms for Directed Graphs}}.
\newblock {\sl \bibinfo{journal}{SIAM J. Comput.}}
  \bibinfo{volume}{37}(\bibinfo{number}{5}), pp. \bibinfo{pages}{1455--1471},
  \doi{10.1137/060650271}.

\bibitemdeclare{techreport}{antom}
\bibitem{antom}
\bibinfo{author}{Tobias Schubert}, \bibinfo{author}{Matthew Lewis} \&
  \bibinfo{author}{Bernd Becker} (\bibinfo{year}{2010}):
  \emph{\bibinfo{title}{Antom: Solver Description}}.
\newblock \bibinfo{type}{Technical Report}, \bibinfo{institution}{{SAT} Race}.

\bibitemdeclare{inproceedings}{tseitin68}
\bibitem{tseitin68}
\bibinfo{author}{G.S. Tseitin} (\bibinfo{year}{1968}): \emph{\bibinfo{title}{On
  the complexity of derivation in propositional calculus}}.
\newblock In \bibinfo{editor}{A.O. Slisenko}, editor: {\sl
  \bibinfo{booktitle}{Structures in Constructive Mathematics and Mathematical
  Logic, Part II, Seminars in Mathematics (translated from Russian)}},
  \bibinfo{publisher}{Steklov Mathematical Institute}, pp.
  \bibinfo{pages}{115--¡V125}.

\bibitemdeclare{inproceedings}{WintersteigerHM09}
\bibitem{WintersteigerHM09}
\bibinfo{author}{Christoph~M. Wintersteiger}, \bibinfo{author}{Youssef Hamadi}
  \& \bibinfo{author}{Leonardo de~Moura} (\bibinfo{year}{2009}):
  \emph{\bibinfo{title}{A Concurrent Portfolio Approach to SMT Solving}}.
\newblock In: {\sl \bibinfo{booktitle}{Proc. of CAV}}, {\sl
  \bibinfo{series}{Lecture Notes in Computer Science}} \bibinfo{volume}{5643},
  \bibinfo{publisher}{Springer}, pp. \bibinfo{pages}{715--720},
  \doi{10.1007/978-3-642-02658-4_60}.

\end{thebibliography}

\end{document}